%% file: main.tex
\documentclass{llncs}

\usepackage{amsmath,amssymb,mathtools,todonotes}

\title{Model Checking Markov Chains as\\ Distribution Transformers}

\author{Rajab Aghamov\inst{1} \and Christel Baier\inst{1} \and Toghrul Karimov\inst{2} \and Joris Nieuwveld \inst{2} \and \\ Joël Ouaknine\inst{2} \and Jakob Piribauer\inst{1} \and Mihir Vahanwala\inst{2}}

\institute{Technische Universität Dresden \and Max Planck Institute for Software Systems, Saarland Informatics Campus}

\input{macros.tex}

\begin{document}
\maketitle

\begin{abstract}
	The conventional perspective on Markov chains considers decision problems concerning the probabilities of temporal properties being satisfied by traces of visited states. However, consider the following query made of a stochastic system modelling the weather: given the conditions today, will there be a day with less than 50\% chance of rain? The conventional perspective is ill-equipped to decide such problems regarding the evolution of the initial distribution. The alternate perspective we consider views Markov chains as \emph{distribution transformers}: the focus is on the sequence of distributions on states at each step, where the evolution is driven by the underlying stochastic transition matrix. More precisely, given an initial distribution vector $\mu$, a stochastic update transition matrix $M$, we ask whether the ensuing sequence of distributions $(\mu, M\mu, M^2\mu, \dots)$ satisfies a given temporal property. This is a special case of the model-checking problem for linear dynamical systems, which is not known to be decidable in full generality.
The goal of this article is to delineate the classes of instances for which this problem can be solved, under the assumption that the dynamics is governed by stochastic matrices.
\end{abstract}

\section{Introduction}
Markov chains are most often regarded by the verification community as a probabilistic means to model the uncertainty inherent to real-world transition systems. We refer the reader to \cite[Chapter 10]{Baier-Katoen} for a thorough exposition and a comprehensive set of references on this perspective. A central application is the verification of network protocols, where it must be certified that packets are transmitted with high probability within a certain number of rounds, or that it is exceedingly unlikely that a queue of requests grows longer than a certain threshold. Here, atomic propositions are framed in terms of states of the system, and events are linear- or branching-time properties. The probability of an event is the measure of the set of runs in which the event occurs, and can be obtained by solving linear programs. The system is deemed correct if events of interest occur with specified probabilities.

The inherent uncertainty plays a more focal role in the study of systems of a more dynamical nature. For instance, consider the weather: it is quite natural to ask whether initialising a forecasting system with the prevalent conditions implies that there will be a day with less than 50\% chance of rain, or whether every week will have a day with less than 25\% chance of rain. The conventional approach does not resolve such queries: here, atomic propositions refer to \emph{distributions} on states and properties specify their evolution in time. For Markov chains, this evolution is governed by the underlying stochastic matrix $M$: if the initial distribution is $\mu$, then the ensuing sequence of distributions is $\mu, M\mu, M^2 \mu, \dots$ ad infinitum.

Research on verifying the evolution of distributions against temporal specifications is not as prevalent as that which adopts the conventional perspective. Decidability in the alternative setting has seemed rather inaccessible: \cite{Agrawal2015,logicprob} only present incomplete or approximate verification procedures, while \cite{KVAK10,KwonGul} owe their model-checking procedures to additional mathematical assumptions. This is not surprising, as the endeavour is fundamentally about solving special instances of the model-checking problem for linear dynamical systems, which is often associated with hard number-theoretic questions.

Formally, a linear dynamical system (LDS) of dimension $k$ is given by an initial vector $v \in \mathbb{Q}^k$ and an update matrix $A \in \mathbb{Q}^{k \times k}$. Its trajectory is the infinite sequence of vectors $(v, Av, A^2v, \dots)$. Let $\mathcal{T} = \{T_1, \dots, T_\ell\}$ be a collection of semialgebraic subsets of $\mathbb{R}^k$, i.e.\ each $T_i$ is defined by a Boolean combination of inequalities involving polynomials with integer coefficients. The characteristic word $\alpha$ of the dynamical system with respect to $\mathcal{T}$ is the infinite word over the alphabet $2^{\mathcal{T}}$ such that for all $n$, $T \in \alpha(n)$ if and only if $A^n v \in T$. The model-checking problem for LDS takes as input $v, A, \mathcal{T}$, and an $\omega$-regular language $\mathcal{L}$ over the alphabet $2^{\mathcal{T}}$ and asks whether the characteristic word $\alpha \in \mathcal{L}$.

Model checking Markov chains as distribution transformers is a special case of the above where $v = \mu$ is a distribution (all entries are non-negative and sum up to $1$), and the update matrix $A = M$ is the (column-)stochastic transition matrix underlying the Markov chain, i.e.\ the $(i, j)$-th entry denotes the transition probability to move from state $j$ to state $i$, and hence each column is a distribution.

The model-checking problem for linear dynamical systems is easily seen to subsume the Skolem problem\footnote{The Skolem problem is often formulated equivalently in terms of linear recurrence sequences (LRS) that satisfy a recurrence relation $u_{n+k} = a_{k-1}u_{n+k-1}+\dots + a_0 u_n$.}, which is a long-standing number-theoretic problem: given an initial vector $v$, update matrix $A$, and normal vector $h$, the Skolem problem asks whether there is an $n$ such that $h^\top A^n v = 0$. The Skolem problem has only been solved for dimension $k \le 4$ \cite{skolem,skolemv}, and is open for $k \ge 5$.
We refer to the Skolem problem in dimension $k$ as {Skolem-$k$}.

Recently, Bilu \emph{et al.\ }\cite{skolembilu} gave a decision procedure for the Skolem problem for diagonalisable $A$; however the guarantee of termination of this procedure is subject to two classical number-theoretic conjectures. A tool that implements the algorithm is available at \cite{skolemtool}.

It is not uncommon for instances of the model-checking problem for linear dynamical systems to be shown decidable by restricting the matrix $A$ to be diagonalisable \cite{prefixindependent,powerpositivity,ouaknine2014positivity} or by bounding its dimension $k$ \cite{ouaknine2014positivity,joeljames3}. The natural question for us is whether restricting the update matrix to be stochastic yields any significant spectral or dimensional benefits that make the model-checking problem tractable. The reductions in \cite{vahanwalaergodic} from the Skolem and closely related problems to the
model-checking problem for Markov chains indicate that the answer is not entirely in the affirmative. Nevertheless, there is some nuance, which we comprehensively explore and detail in this article.

\subsection*{Our contributions}
We identify that the dynamics of Markov chains relevant to our model-checking problem play out in a space of dimension lower than that of the ambient space. The following definition explicitly captures the observation that if $\mu$ is a distribution and $M$ is stochastic, then each element of the sequence $(\mu, M\mu, M^2 \mu, \dots)$ is also a distribution on the probability simplex $\Delta = \{x: x \ge 0 \land x_1 + \dots + x_k = 1\}$, which is contained in an affine subspace.
\begin{definition}
	The dynamical dimension of a stochastic matrix $M$ is defined as the number of nonzero eigenvalues (counted with algebraic multiplicity) of $M$ that have modulus strictly less than $1$.
\end{definition}

Technically speaking, it is the maximum dimension of the semialgebraic sets involved (see Sec.~\ref{semialgebraic prelims}) that plays a crucial role in characterising the solvability of the model-checking problem for linear dynamical systems. Semialgebraic sets come equipped with a notion of \emph{intrinsic dimension} (intuitively, in ambient $3$-dimensional space, a point has intrinsic dimension $0$, a wire has intrinsic dimension $1$, a surface has intrinsic dimension $2$, and a solid has intrinsic dimension $3$). There is also a notion of \emph{linear dimension}: the dimension of the smallest linear subspace that contains the set in question.

To solve the model-checking problem for Markov chains of dynamical dimension $k$, we project the problem onto an instance of a linear dynamical system of ambient dimension $k$. It is clear that the attendant semialgebraic sets become at most $k$-dimensional in doing so. We show additionally that the linear dimension necessarily decreases as well:

\begin{theorem}[Auxiliary decidability result]
	\label{aux}
	The model-checking problem for Markov chains with instances $(\mu, M, \mathcal{T}, \mathcal{L})$ such that
	\begin{itemize}
		\item[(a)]  $M$ has dynamical dimension $k$, and
		\item[(b)]  the target sets in $\mathcal{T}$ have intrinsic dimension at most $d_1$ and linear dimension at most $d_2$
	\end{itemize}
	reduces to the model-checking problem for linear dynamical systems with instances $(v, A, \mathcal{T}', \mathcal{L}')$ such that
	\begin{itemize}
		\item[(A)] $A$ is an invertible $k \times k$ matrix, and
		\item[(B)] the target sets in $\mathcal{T}'$ have intrinsic dimension at most $d_1$ and linear dimension at most $d_2-1$.
	\end{itemize}
\end{theorem}

We remark that proving the above requires us to preserve not only the spectral properties of the update matrix but also the rationality of the matrix and the initial vector.

The state of the art motivates us to define a criterion for semialgebraic sets to be considered low-dimensional with respect to our model-checking problem for Markov chains.

\begin{theorem}[Main result of \cite{tamedecidable}]
\label{tamedecidability}
The model-checking problem $(v, A, \mathcal{T}, \mathcal{L})$ for linear dynamical systems with $\mathcal{T}$ being a collection of semialgebraic sets whose intrinsic dimension is at most $1$ or linear dimension is at most $3$ is decidable.
\end{theorem}

Note that the above result is indeed at the cutting edge of decidability: the open Skolem problem in dimension $5$, for example, queries the reachability of a set that has linear dimension $4$.

\begin{definition}
	A semialgebraic set $T \subseteq \mathbb{R}^k$ is said to be Markov-low-dimensional if it has intrinsic dimension at most $1$ or is contained in a linear subspace of dimension at most $4$.
\end{definition}

Combining Theorems~\ref{aux} and~\ref{tamedecidability} gives us our main decidability result.

\begin{theorem}[Main decidability result]
	\label{decidability}
	The model-checking problem for Markov chains restricted to instances $(\mu, M, \mathcal{T}, \mathcal{L})$ such that either
	\begin{itemize}
		\item[(a)] $\mathcal{T}$ is a collection of Markov-low-dimensional sets, or
		\item[(b)]the dynamical dimension of $M$ is at most~$3$
	\end{itemize}
	is decidable.
\end{theorem}

For instances in which both hypotheses (a) and (b) of Theorem~\ref{decidability} fail to hold, we establish hardness through a slightly generalised version of the result of \cite{vahanwalaergodic}. Recall that semialgebraic sets are defined in terms of polynomial inequalities. We call a set homogeneous if all the polynomials involved are homogeneous, and $s$-homogeneous if all the polynomials involved would be homogeneous were the origin shifted to $s$. We refer the reader to Sec.\ \ref{semialgebraic prelims} for detailed definitions of these terms.

\begin{theorem}[First hardness result]
	\label{ldstomc}
	Let $s \in \mathbb{Q}^{k+1}$ be any distribution with strictly positive entries.
	The model-checking problem for $k$-dimensional linear dynamical systems with homogeneous target sets reduces to the model-checking problem for order $k+1$ ergodic Markov chains with $s$-homogeneous target sets.
\end{theorem}

\begin{theorem}[Second Hardness Result]
	\label{skolemhard}
	 Skolem-5 Turing-reduces to the reachability problem for ergodic Markov chains and semialgebraic targets of intrinsic dimension at most~2.
\end{theorem}

We give the prerequisite mathematical background in Sec.\ \ref{background}, prove our decidability results in Sec.\ \ref{decidability proofs}, our hardness results in Sec.\ \ref{hardness proofs}, and finally offer concluding remarks in Sec.\ \ref{conclusion}.

\section{Mathematical background}
\label{background}
We will use $\mathbf{0}_k$ to denote the zero vector of dimension~$k$, and $\mathbf{1}_k$ to denote the $k$-dimensional vector whose entries are all $1$.
If $v_i \in \rel^{d_i}$ for $1\le i \le k$, then by $(v_1,\ldots,v_k)$ we mean the concatenation of $v_1,\ldots,v_k$ belonging to $\rel^{d_1+\ldots+d_k}$.

\subsection{Automata over infinite words}
Let $\Sigma$ be a finite non-empty alphabet. We use $\Sigma^*, \Sigma^\omega$ to respectively denote the sets of finite and infinite words over $\Sigma$. An $\omega$-language $\mathcal{L}$ over $\Sigma$ is a subset of $\Sigma^\omega$. A language $\mathcal{L}$ is $\omega$-regular if and only if it is accepted by a deterministic Muller automaton, which is defined as follows.

\begin{definition}
A deterministic Muller automaton $\mathcal{A}$ is given by $ (Q, \Sigma, q_0, \delta, \mathcal{F})$ where $Q$ is a finite set of states, $\Sigma$ is a finite alphabet, $q_0$ is the initial state, $\delta : Q \times \Sigma \rightarrow Q$ is the transition function, and $\mathcal{F} = \{F_1, \dots, F_a\} \subseteq 2^{Q}$. Let $\text{Inf}(\alpha)$ denote the set of states visited infinitely often by the run of $\mathcal{A}$ on $\alpha$. The automaton accepts $\alpha$ if $\text{Inf}(\alpha) \in \mathcal{F}$.
\end{definition}

\subsection{Semialgebraic sets}
\label{semialgebraic prelims}
\begin{definition}
A semialgebraic set $T \subseteq \mathbb{R}^k$ is expressed as a Boolean combination of finitely many polynomial (in)equalities $p_i(x_1, \dots, x_k) \sim 0$, where each $p_i \in \mathbb{Z}[x_1, \dots, x_k]$ (i.e.\ has integer coefficients). Then $T$ is defined as the set of all $x \in \mathbb{R}$ which satisfy the combination of (in)equalities.
\end{definition}

\begin{definition}
A semialgebraic set $T \subseteq \mathbb{R}^k$ is said to be homogeneous if it can be expressed as a Boolean combination of finitely many polynomial (in)equalities $p_i(x_1, \dots, x_k) \sim 0$, where each $p_i$ is homogeneous, i.e.\ all its additive monomial terms have the same degree. If there exists $s\in \mathbb{R}^k$ such that each $p_i(x - s)$ is homogeneous, then $S$ is said to be $s$-homogeneous.
\end{definition}

By definition, we have that for all $\lambda > 0$, and nonzero vectors $x$, we have $x \in T$ if and only if $\lambda x \in T$ for homogeneous sets $T$, and $s+ x \in T'$ if and only if $s + \lambda x \in T'$ for $s$-homogeneous sets $T'$.

Using cell decomposition (see \cite[Chapter 5]{algebraicgeometrybook}), a semialgebraic set $T$ can be decomposed into a finite union of semialgebraic sets $T_1, \dots, T_\ell$, where each $T_i$ is (semialgebraically) homeomorphic to a single point or $(0, 1)^{d_i}$ for some $d_i \ge 1$. The intrinsic dimension $d$ of $T$ is defined to be $\dim(T) = \max_i d_i$;
As \cite[Chapter 5.3]{algebraicgeometrybook} shows, $d$ is independent of the choice of decomposition and hence well-defined. The linear dimension of $T$ is the dimension of the smallest linear subspace of $\mathbb{R}^k$ that contains $T$.



Next we define a \emph{Zariski topology} on $\rel^d$.
An \emph{algebraic} set $X$ is defined by $p(x_1,\ldots,x_d) = 0$ where $p \in \rat[x_1,\ldots,x_d]$.
Note that a conjunction
\[
\bigwedge_{i\in I} p_i(x_1,\ldots,x_d) = 0
\]
of polynomial equalities can be defined by a single polynomial equality
\[
\sum_{i\in I} p_i(x_1,\ldots,x_l)^2 = 0.
\]
The closed sets in the Zariski topology are exactly the algebraic sets.

Let $f \st \rel^d \to \rel^k$ be given by $f(x) = (p_1(x)/q_1(x),\ldots,p_k(x)/q_k(x))$ where each $p_i,q_i \in \rat[x_1,\ldots,x_d]$.
Suppose $q_i(x) \ne 0$ for all $x \in \rel^d$ and $1\le i \le k$.
Consider a Zariski-closed set $X \subseteq \rel^d$ defined by a polynomial $h$.
Then $f^{-1}(X)$ is exactly the set defined by $h(p_1(x)q_2(x)\cdots q_k(x),\ldots,p_k(x)q_1(x)\cdots q_{k-1}(x)) = 0$, which is Zariski-closed.
It follows that such a function $f$ is continuous in the Zariski topology.

Given a topology, a set $X$ is irreducible if its closure $\overline{X}$ cannot be written as a union of two closed sets different from $\overline{X}$.
Irreducibility is preserved under continuous mappings.

\subsection{Linear Algebra}
We begin by recalling some standard terminology and results. The dimension of a vector space is the smallest number of (necessarily linearly independent) vectors required to span it. Given a matrix $A$, its rank is equivalently: (a) the size of the largest subset of linearly independent rows (dimension of row space); (b) the size of the largest subset of linearly independent columns (dimension of column space). The kernel of an $m\times n$ matrix $A$ is the vector space of solutions to $Ax = \mathbf{0}_m$. The rank-nullity theorem states that the rank of a matrix plus the dimension of its kernel equals the number of columns.

We follow a standard text \cite[Chapter 8]{axler} in recording some spectral properties of matrices.
\begin{definition}
Let $A \in \mathbb{C}^{k\times k}$. If $\lambda \in \mathbb{C}$ satisfies $\det(A - \lambda I) = 0$, it is called an eigenvalue of $A$. A nonzero vector $v \in \mathbb{C}^k$ that satisfies $(A - \lambda I)v = \mathbf{0}_k$ is called a (right) eigenvector of $\lambda$. A nonzero vector $v' \in \mathbb{C}^k$ that satisfies $v'^\top (A - \lambda I) = \mathbf{0}_k^\top$ is called a left eigenvector of $\lambda$.
\end{definition}

If not explicitly specified, an eigenvector means a right eigenvector. We seek to generalise the definition of eigenvectors; the following definition can be made analogously for left eigenvectors too. If not explicitly specified, an eigenvector means an order-$1$ eigenvector in the sense of the following definition.

\begin{definition}
A vector $v$ that satisfies $(A - \lambda I )^j v = \mathbf{0}_k$, but $(A - \lambda I)^{j-1} v \ne \mathbf{0}_k$, is called an order-$j$ generalised eigenvector of $\lambda$.
\end{definition}

An eigenvalue $\lambda$ is called \emph{simple} if all its generalised eigenvectors have order $1$. Note that $\lambda$ has a left order-$(j+1)$ generalised eigenvector if and only if it has a right order-$(j+1)$ generalised eigenvector if and only if $\text{rank}(A - \lambda I)^{j+1} < \text{rank}(A - \lambda I)^j$.

By definition, we observe that if $v_{j+1}$ is an order-$(j+1)$ generalised eigenvector, then $Av_{j+1}$ is of the form $\lambda v_{j+1} + v_j$, where $v_j \in \ker(A - \lambda I)^j$ but $v_j \notin \ker(A - \lambda I)^{j-1}$, i.e.\ $v_j = Av_{j+1} - \lambda v_{j+1}$ is an order-$j$ generalised eigenvector. Given $v_j$ of order $j$, let $v_j, \dots, v_1$ be a chain of generalised eigenvectors of $\lambda$ obtained thus. We can show, by a simple induction on $n$, that:
\begin{equation}
\label{powers}
A^n v_j = \lambda^n v_j + \binom{n}{1}\lambda^{n-1} v_{j-1} + \dots + \binom{n}{j-1}\lambda^{n-j+1} v_1.
\end{equation}

We immediately observe the invariance property of the following result \cite[Construction 8.20]{axler}:
\begin{lemma}
\label{invariant}
The eigenspace $V \subseteq \mathbb{C}^k$ of eigenvalue $\lambda$ is the vector space spanned by the generalised eigenvectors of the eigenvalue $\lambda$ of $A$. The dimension $\ell$ of $V$ is equal to the algebraic multiplicity of the eigenvalue $\lambda$. The space $V$ is the kernel of $(A - \lambda I)^k$, and a basis can be expressed accordingly. Moreover, $V$ is invariant under $A$, i.e.\ if $v \in V$, then $Av \in V$.
\end{lemma}

In order to state our next property, we need a notion of composition of vector spaces.
\begin{definition}[Sum of vector spaces]
\label{directsum}
Let $V_1, \dots, V_m$ be linear subspaces of $\mathbb{C}^k$ with linearly independent bases $B_1, \dots, B_m$ respectively. Their sum $V = V_1 + \dots + V_m$ is the vector space spanned by the set of vectors $B = B_1 \cup \dots \cup B_m$. Furthermore, if the set $B$ is linearly independent, then the sum is called a direct sum, and is denoted $V_1 \oplus \dots \oplus V_m$.
\end{definition}

By definition, the dimension of the direct sum of vector spaces is equal to the sum of the dimensions of the vector spaces. The definition of direct sum also implies that $\mathbf{0}_k$ is the only vector common to any pair of subspaces being summed.

Eigenspaces permit a convenient decomposition of $\mathbb{C}$ \cite[Theorem 8.22]{axler}: they are linearly independent and span the entire space.
\begin{theorem}
\label{decomposition}
Let $\lambda_1, \dots, \lambda_m$ be distinct eigenvalues of $A \in \mathbb{C}^{k \times k}$, with respective eigenspaces $V_1, \dots, V_m$. Then $V_1 \oplus \dots \oplus V_m = \mathbb{C}^k$.
\end{theorem}

We note the following property.
\begin{lemma}
\label{ortho}
Let $\lambda, \eta$ be distinct eigenvalues of $A \in \mathbb{C}^{k \times k}$. Let $v_\eta$ be a generalised left eigenvector of $\eta$ and let $v_\lambda$ be a generalised right eigenvector of $\lambda$. Then $v_\eta^\top v_\lambda = 0$.
\end{lemma}
\begin{proof}
We prove this by induction on the orders $i, j$ of the generalised left and right eigenvectors $v_\eta, v_\lambda$ respectively. We have that $(i, j) > (i', j')$ if $i \ge i'$ and $j \ge j'$ with at least one inequality being strict.

We prove the base case (where $v_\eta, v_\lambda$ are order-$1$). Observe by associativity that $v_\eta^\top A v_\lambda = \eta v_\eta^\top v_\lambda = \lambda v_\eta^\top v_\lambda$. Since $\lambda, \eta$ are distinct, it must be that $v_\eta^\top v_\lambda = 0$.

Now for the induction step (where $v_\eta, v_\lambda$ are of order $i, j$ respectively), we assume that $(v'_\eta)^\top v'_\lambda = 0$ for all $v'_\eta, v'_\lambda$ of orders $i', j'$, with $(i, j) > (i', j')$. Observe by associativity that $v_\eta^\top A v_\lambda = (\eta v_\eta + v'_\eta)^\top v_\lambda =  v_\eta^\top (\lambda v_\lambda + v'_\lambda)$.

We have that $v'_\eta$ is of order $i-1$, and $v'_\lambda$ is of order $j-1$ \footnote{If either of $i, j$ is $1$, then the proof is even simpler because the corresponding $v'$ is the zero vector}. Thus by the induction hypothesis, $(v'_\eta)^\top v_\lambda = v_\eta^\top v'_\lambda = 0$. Substituting into the above equality, we get that $v_\eta^\top A v_\lambda = \eta v_\eta^\top v_\lambda = \lambda v_\eta^\top v_\lambda$ as before:  it must be that $v_\eta^\top v_\lambda = 0$.
\qed
\end{proof}

\subsection{Rational Stochastic Matrices}
We use rational (column-)stochastic matrices to specify Markov chains. We shall apply linear-algebraic results to state their important properties.
\begin{definition}
A matrix $M \in \mathbb{Q}^{k \times k}$ is said to be stochastic if each entry is non-negative, and $\mathbf{1}_k^\top M = \mathbf{1}_k^\top$, i.e. the entries of each column sum up to $1$.
\end{definition}

Very clearly, $1$ is an eigenvalue of $M$: $\det(M - I) = 0$ since by definition, $\mathbf{1}_k^\top(M - I) = \mathbf{0}_k^\top$. The following property holds.

\begin{lemma}
\label{eigenatmostone}
Let $\lambda$ be an eigenvalue of a stochastic matrix $M$. We have that $|\lambda| \le 1$. Moreover, if $|\lambda | = 1$, then $\lambda$ must be a root of unity and a simple eigenvalue (all generalised eigenvectors are of order $1$).
\end{lemma}
\begin{proof}
Consider a left eigenvector $v = (v_1, \dots, v_k)$ of $\lambda$ chosen such that the entry of maximum modulus has modulus $1$, and let $T \subset \mathbb{C}$ be the bounded polytope defined as the convex hull of the points $v_1, \dots, v_k \in \mathbb{C}$. Clearly, $T$ is contained in the unit circle centred at the origin. Now, since the entries of each column of $M$ are non-negative and sum up to $1$, each coordinate of $v^\top M = (\lambda v_1, \dots, \lambda v_k)$ is a convex combination of $v_1, \dots, v_k$, and is hence contained in $T$, and thus the unit circle. It follows that $|\lambda| \le 1$.

Moreover, observe that $T$ can only intersect the unit circle at its corners. Now observe that for all $n\ge 1$, $\lambda^n v^\top = v^\top M^n$, and that the entries of each column of $M^n$ are non-negative and sum up to $1$. Thus, by the same arguments as above, each entry of $(\lambda^n v_1, \dots, \lambda^n v_k)$ must be contained in $T$. Suppose $|v_i| = 1$. Then each $\lambda^n v_i$ must be one of the finitely many corners of the convex hull of $v_1, \dots, v_k$: this is only possible if $\lambda$ is a root of unity.

Now assume that an eigenvalue $\lambda$ such that $\lambda^d = 1$ were not simple. Then there is an order-$2$ left eigenvector $v = (v_1, \dots, v_k)$ of $\lambda$ chosen such that one of the entries $v_i$ with maximum modulus is $1$. Define $T$ as before, and note for all $n$ that each coordinate of $v^\top M^n$ must lie in $T$ and hence the unit circle. However, by equation \ref{powers} we have
$
v^\top M^{nd} = v^\top + nd (v')^\top
$.
The coordinates of $v^\top M^{nd} $ clearly cannot lie in the unit circle for arbitrarily large $n$: a contradiction.
\qed
\end{proof}

\section{Proofs of the decidability results}
\label{decidability proofs}
In this section, we establish our decidability results. We start with an instance $(\mu, M, \mathcal{T}, \mathcal{L})$ of the model-checking problem for Markov chains, where $M$ has dynamical dimension $d_1$, the sets of $\mathcal{T}$ have linear dimension at most $d_2$ and intrinsic dimension at most $d_3$. We show how to reduce this instance to an instance $(v, A, \mathcal{T}', \mathcal{L'})$ of the model-checking problem for linear dynamical systems, where $A$ is an invertible $d_1 \times d_1$ matrix, the sets of $\mathcal{T}$ have linear dimension at most $\min(d_1, d_2-1)$ and intrinsic dimension at most $\min(d_1, d_3)$. This would precisely be the proof of Theorem~\ref{aux}. Given the state-of-the-art Theorem~\ref{tamedecidability}, this implies that instances $(\mu, M, \mathcal{T}, \mathcal{L})$ of the form where $M$ has dynamical dimension at most $3$ or the sets of $\mathcal{T}$ have linear dimension at most $4$ or intrinsic dimension at most $1$ are decidable, thus proving Theorem~\ref{decidability}.

We perform the reduction to prove Theorem~\ref{aux} through three lemmata.
\begin{enumerate}
\item Lemma~\ref{zeroelim} implies that it suffices to consider linear dynamical systems with invertible update matrices.
\item Lemma~\ref{rootelim} implies that assuming the update matrix is non-degenerate\footnote{
A matrix $A$ is said to be \emph{non-degenerate} if for every pair of distinct eigenvalues $\lambda, \lambda'$, the ratio $\lambda/\lambda'$ is not a root of unity.} does not lose any generality.
\item Lemma \ref{mctolds} takes as input an instance of a linear dynamical system model-checking problem obtained upon performing the preprocessing of Lemmata \ref{zeroelim} and \ref{rootelim} on $(\mu, M, \mathcal{T}, \mathcal{L})$. It reduces its input into an instance of the model-checking problem with a $d_1 \times d_1$ invertible matrix and semialgebraic sets of lower dimensions.
\end{enumerate}

We remark that the lemmata are elementary, but not immediate, as they preserve both the rationality and the spectra of the involved matrices.

We start with an instance $(\mu, M, \mathcal{T}, \mathcal{L})$, where $M$ is a stochastic matrix, and the eigenvalues of $M$ are $1, \nu_1, \dots, \nu_{j'}, \gamma_1, \dots, \gamma_j, $ and possibly $0$. By Lemma \ref{eigenatmostone}, the eigenvalues $\nu_1, \dots, \nu_{j'}$ of modulus $1$ must be roots of unity, and the eigenvalues $\gamma_1, \dots, \gamma_j$ have modulus less than $1$. Furthermore, the eigenvalues $1, \nu_1, \dots, \nu_{j'}$ must be simple, i.e.\ their generalised eigenvectors have order $1$. We let $\ell_1$ be the multiplicity of the eigenvalue $0$. The following lemma shows that we can obtain a dynamical system with the same spectral properties as stated above sans the eigenvalue $0$.
\begin{lemma}
\label{zeroelim}
Let $\mu \in \mathbb{Q}^{k+\ell}$, and let $M \in \mathbb{Q}^{(k + \ell)\times (k+\ell)}$ be a matrix with distinct eigenvalues $0, \lambda_1, \dots, \lambda_j$ with the eigenvalue $0$ having (algebraic) multiplicity $\ell$. We can compute a rank-$k$ matrix $Q \in \mathbb{Q}^{(k+\ell)\times k}$, an invertible matrix $B \in \mathbb{Q}^{k \times k}$ and $v \in \mathbb{Q}^k$ such that:
\begin{enumerate}
\item For all $n \ge \ell$, we have $M^n \mu = QB^n v$.
\item Let $n \ge \ell$. For any semialgebraic set $T \subseteq \mathbb{R}^{k +\ell}$, $M^n \mu \in T$ if and only if $B^n v \in T'$, where $T' = \{x \in \mathbb{R}^k: Qx \in T\}$ and has intrinsic and linear dimensions at most those of $T$.
\item The eigenvalues of $B$ are $\lambda_1, \dots, \lambda_j$. Moreover, for any $i \ge 0$ and $\lambda \in \mathbb{C}_{\ne 0}$, $y$ is a generalised order-$i$ eigenvector of $B$ corresponding to $\lambda$ if and only if $Qy$ is a generalised order-$i$ eigenvector of $M$ corresponding to $\lambda$.
\end{enumerate}
\end{lemma}
\begin{proof}
Let $V_0, V_{\lambda_1}, \dots, V_{\lambda_j}$ denote the eigenspaces of the respective eigenvalues. Recall that by Theorem~\ref{decomposition} we can define the vector space $W =  V_{\lambda_1} \oplus \cdots \oplus V_{\lambda_j}$ as a direct sum, and that $V_0 \oplus W = \mathbb{C}^k$. By Lemma \ref{invariant} and the properties implied by Def.\ \ref{directsum}, we have that the dimension of $V_0$ (over $\mathbb{C}$) is $\ell$, $V_0 \cap W = \{\mathbf{0}_{k+\ell}\}$, and that the dimension of $W$ is $k$.

The decomposition into $V_0$ and $W$ is the key ingredient of the proof because these are invariant subspaces. Indeed, if $z \in V_0$ then by Equation (\ref{powers}) we have $M^\ell z = \mathbf{0}_{k+\ell}$. For a nonzero $w \in W$, we apply Lemma \ref{invariant} to each of the components $V_\lambda$ of $W$ to conclude that $Mw \in W$ and that $Mw$ is moreover nonzero since $w$ is not in the eigenspace of $0$.

Rational bases for $V_0$ and $W$ can effectively be computed. A rational basis for $V_0$ can be obtained by solving the rational homogeneous linear system $M^\ell x = \mathbf{0}_{k+\ell}$, one for $W$ can be obtained by taking $k$ linearly independent columns of $M^\ell$. To argue the latter, recall that by Lemma \ref{ortho}, if $v'$ is in the left eigenspace $V_0'$ of $0$ and $w \in W$, then $(v')^\top w = 0$. A rational basis for $V_0'$ can be obtained similarly to the one for $V_0$: let its $\ell$ vectors be the rows of $H$. By the above observation, $W$, which has dimension $k$ over $\mathbb{C}$, is also contained in the kernel of $H$, which, by the rank-nullity theorem, has dimension $k$. It follows that $W$ is indeed the kernel of $H$, and a rational basis can be obtained by solving a rational homogeneous linear system.

We take $P, Q$ to respectively be the matrices whose columns form the thusly computed rational bases of $V_0, W$.  Note that the matrix $R = \begin{bmatrix}P & Q\end{bmatrix}$ is invertible. We take $P', Q'$ such that $R^{-1} = \begin{bmatrix}P' \\ Q'\end{bmatrix}$. This choice establishes that $P', Q', P, Q$ have full rank. In particular, we remark that $Q$ defines an invertible map from $\mathbb{R}^k$ to $W$, and is inverted by $Q'$, since indeed $Q'Q = I$.

Now, construct $D \in \mathbb{Q}^{(k+\ell)\times (k+\ell)}$ as
\begin{align*}
D &= R^{-1}M R = \begin{bmatrix} P' \\ Q' \end{bmatrix}M\begin{bmatrix}P & Q\end{bmatrix} \\
&= \begin{bmatrix} P' \\ Q' \end{bmatrix}\begin{bmatrix}MP & MQ\end{bmatrix}\\
&= \begin{bmatrix}
P'MP & P'MQ \\
O & Q'MQ
\end{bmatrix} \text{ (by definition of $R^{-1}$)} \\
 &= \begin{bmatrix}
P'MP & O \\
O & Q'MQ
\end{bmatrix} \text{ (since each column of $MQ$ is still in $W$)} \\
 &= \begin{bmatrix}
P'MP & O \\
O & B
\end{bmatrix} \text{ (defining the bottom right $k\times k$ matrix $B = Q'MQ$)}.
\end{align*}

We observe that for $n \ge \ell$
$$
D^n = R^{-1}M^n R = \begin{bmatrix}
P'M^nP & O \\
O & B^n
\end{bmatrix} = \begin{bmatrix}
O & O \\
O & B^n
\end{bmatrix}
$$
since $M^\ell P = O$ because $P$ spans the eigenspace of $0$.

We can compute the unique $u \in \mathbb{Q}^\ell, v \in \mathbb{Q}^k$ such that $\mu = Pu + Qv$. We now observe that for $n \ge \ell$
\begin{align*}
M^n \mu &= RD^nR^{-1}\mu = RD^nR^{-1}R\begin{bmatrix}u \\ v\end{bmatrix} \\
&= R\begin{bmatrix}
O & O \\
O & B^n
\end{bmatrix}
\begin{bmatrix}u \\ v\end{bmatrix} = \begin{bmatrix}P & Q\end{bmatrix}\begin{bmatrix}\mathbf{0}_\ell \\ B^nv\end{bmatrix}
= QB^n v.
\end{align*}
This establishes requirement (1) of the statement.

It is obvious that for $n \ge \ell$, $M^n \mu \in T$ if and only if $B^n v \in  \{x: Qx \in T\}$. We observe that for $n \ge \ell$, $M^n \mu \in W$. Thus, $M^n \mu \in T$ if and only if $M^n \mu \in T \cap W$. Since $W$ is a linear subspace, the linear dimension of $T \cap W \subseteq T$ is at most the linear dimension of $T$, and the intrinsic dimension also cannot increase \cite[Proposition 5.28]{algebraicgeometrybook}. Now, recall that $Q'$ is an invertible linear map from $W$ to $\mathbb{R}^k$, and hence a homeomorphism from $T \cap W$ to $T'$. The linear and intrinsic dimensions of $T'$ are hence the same as those of $T \cap W$. This establishes requirement (2) of the statement.

Finally, we prove that $B = Q'MQ$ has almost the same spectral properties as that of $M$. We recall that $Q$ is a linear map from $\mathbb{R}^k$ to $W$, the space $W$ is invariant under $M$ and moreover that $M$ never maps a nonzero $w \in W$ to the zero vector, and $Q'$ inverts $Q$. This already establishes that $Bx \ne \mathbf{0}$ for nonzero $x$, and hence the invertibility of $B$. We shall work with pairs $y \in \mathbb{C}^k, w \in W$ such that $w = Qy, y = Q'w$.

Now, $By = Q'MQy = Q'Mw$. Also, $\lambda y = \lambda Q'w$. If $By = \lambda y$, we also have $Q'(Mw) = Q'(\lambda w)$. The images of two nonzero vectors in $W$ under the invertible $Q'$ are equal: hence the vectors must be equal. We thus have $y$ is an eigenvector of $B$ with eigenvalue $\lambda$ only if $w = Qy$ is an eigenvector of $M$ with eigenvalue $\lambda$.

Conversely, $Q'Mw = Q'MQy = By$, and $Q'(\lambda w) = Q'Q (\lambda y) = \lambda y$. If $Mw = \lambda w$, we must also have $By = \lambda y$. This establishes that $B$ has the same set of eigenvectors as $M$ (for $\lambda\ne 0$), with $Q'$ mapping the eigenvectors of $M$ to those of $B$.

We can apply the above reasoning in a straightforward induction on the order of generalised eigenvectors to prove that this mapping holds for the entire eigenspaces of each eigenvalue. This establishes requirement (3) of the statement, and completes the proof of Lemma \ref{zeroelim}.
\qed
\end{proof}

Let $\alpha$ be the characteristic word of the dynamical system $(M, \mu)$ with respect to $T$. Take $\mathcal{L}_1 =\{\beta: \alpha(0)\alpha(1)\cdots\alpha(\text{mult}_0-1)\cdot \beta \in \mathcal{L}\} $. It is clear that $\mathcal{L}$ is $\omega$-regular. Lemma \ref{zeroelim} thus reduces $(\mu, M, \mathcal{T}, \mathcal{L})$ to $(\mu_1, M_1, \mathcal{T}_1, \mathcal{L}_1)$, where $M_1$ is invertible, and whose eigenvalues of modulus $1$ are all simple and $c$-th roots of unity. We also have that the dimensions of sets in $\mathcal{T}_1$ are at most those in $\mathcal{T}$. We now convert all eigenvalues that are roots of unity to $1$.

\begin{lemma}
\label{rootelim}
Let $v \in \mathbb{Q}^k$, let $B \in \mathbb{Q}^{k \times k}$ be an invertible matrix with eigenvalues $\lambda_1, \dots, \lambda_j$, and let $c \in \mathbb{N} \backslash\{0\}$. Let $\mathcal{T} = \{T_1, \dots, T_h\}$ be a collection of semialgebraic sets, and let $\mathcal{L}$ be an $\omega$-regular language over $2^{\mathcal{T}}$. The following hold:
\begin{enumerate}
\item The matrix $B^c$ has eigenvalues $\lambda_1^c, \dots, \lambda_j^c$. If $v$ is an order-$i$ generalised eigenvector of some eigenvalue $\lambda$ of $B$, then it is also an order-$i$ generalised eigenvector of eigenvalue $\lambda^c$ of $B^c$, and conversely.
\item We can construct $\mathcal{T}'$ and an $\omega$-regular language $\mathcal{L'}$ over $2^{\mathcal{T}'}$ such that the characteristic word $\alpha$ of $B, v$ with respect to $\mathcal{T}$ is in $\mathcal{L}$ if and only if the characteristic word $\alpha'$ of $B^c, v$ with respect to $\mathcal{T'}$ is in $\mathcal{L}'$.
\end{enumerate}
\end{lemma}
\begin{proof}
It is immediately clear that $B^c$ has $\lambda_1^c, \dots, \lambda_j^c$ among its eigenvalues. We prove the retention of generalised eigenvectors by induction on the order $i$. The base case $i = 1$ is trivial: $Bv = \lambda v$ implies $B^cv = \lambda^c v$. For the induction step, suppose the claim holds for $i$, and let $v_{i+1}$ be an order-$(i+1)$ generalised vector of $\lambda$ of $B$. By Equation (\ref{powers}), $B^c v_{i+1} = \lambda^c v_{i+1} + (f_i v_i + \dots + f_{i-1}v_1)$. The latter summand, by the induction hypothesis, is an order-$i$ generalised eigenvector of $\lambda^c$, making $v_{i+1}$ an order-$(i+1)$ generalised eigenvector.

Since we have identified $k$ generalised eigenvectors of $B^c$, there can be no more. We can also use Lemma \ref{invariant} to conclude that since a basis of eigenvectors of $B$ is also a basis of eigenvectors of $B^c$, the multiplicities of $\lambda_1^c, \dots, \lambda_j^c$ add up to $k$, and there are no other eigenvalues of $B^c$. This establishes property (1).

The key observation to prove property (2) is that $B^{qc + r}v = B^r (B^c)^q v$. Thus, for any $n = qc + r$ and semialgebraic set $T$, we have $B^n v \in T$ if and only if $(B^c)^q v \in T_r = \{x: B^r x \in T\}$. Note that $T_r$ is the pre-image of $T$ under $B^r$. Since $B$ is invertible, the map from $T$ to $T_r$ is a homeomorphism, and both sets have the same dimensions. Define $\mathcal{T}' = \{T_{1,0}, \dots, T_{h, 0}, \dots, T_{1, c-1}, \dots, T_{h, c-1}\}$ where $T_{i, r} = \{x : B^{r}x \in T_i\}$.

It is now easy to observe that the characteristic word $\alpha$ of $B, v$ is a ``flattening'' of $\alpha'$ of $B^c, v$. Formally, $T_i \in \alpha(qc + r)$ if and only if $T_{i, r} \in \alpha'(q)$. Define $\varphi: 2^{\mathcal{T}'} \rightarrow \left(2^{\mathcal{T}}\right)^c$ as $\varphi(\sigma') = \sigma_0\dots\sigma_{c-1}$, where $T_i \in \sigma_r$ if and only if $T_{i, r} \in \Sigma$. Define the extension $\varphi': \left(2^{\mathcal{T}'}\right)^\omega \rightarrow  \left(2^{\mathcal{T}}\right)^\omega$ as $\varphi'(\alpha') = \varphi(\alpha'(0))\varphi(\alpha'(1)) \cdots$ and observe that $\varphi'$ is bijective. Define $\mathcal{L'} = \{\alpha': \varphi'(\alpha') \in \mathcal{L}\}$. Clearly, $\alpha' \in \mathcal{L}'$ if and only if $\alpha \in \mathcal{L}$. In only remains to show that $\mathcal{L}'$ is $\omega$-regular.

We shall use a deterministic Muller automaton $\mathcal{A} = (Q, 2^{\mathcal{T}}, q_0, \delta, \mathcal{F})$ that accepts $\mathcal{L}$ to give a deterministic Muller automaton $\mathcal{A}' = (Q, 2^{\mathcal{T}'}, q_0', \delta', F')$ that accepts $\mathcal{L'}$ in order to do so. Intuitively, $\mathcal{A}'$, upon reading each letter $\sigma'$ of $\alpha'$ simulates what $\mathcal{A}$ would do upon reading the corresponding $\varphi(\sigma')$.
\begin{align*}
Q &= Q \times 2^Q \\
q_0 &= (q_0, \{\}) \\
\delta'((q, S_{prev}), \sigma') &= (\delta(q, \varphi(\sigma')), S_{new})
\end{align*}
where $S_{new}$ is the set of states traversed by $\mathcal{A}$ while reading $\varphi(\sigma')$ and going from $q$ to $\delta(q, \varphi(\sigma'))$ along a path of length $c$. We have that
$$
\text{Inf}_{\mathcal{A}}(\alpha) = \bigcup_{(q, S) \in \text{Inf}_{\mathcal{A}'}(\alpha')} S.
$$
Thus the acceptance condition
$$
F' \in \mathcal{F}' \Leftrightarrow \left(\bigcup_{(q, S)\in F' }S\right) \in \mathcal{F}.
$$

This establishes property (2) and completes the proof of Lemma \ref{rootelim}.
\qed
\end{proof}

Lemma \ref{rootelim} thus reduces $(\mu_1, M_1, \mathcal{T}_1, \mathcal{L}_1)$ to $(\mu_1, M_2, \mathcal{T}_2, \mathcal{L}_2)$, where $M_2$ is invertible, and whose only eigenvalue of unit modulus is simple and equal to $1$. We also have that the dimensions of sets in $\mathcal{T}_2$ are at most those in $\mathcal{T}$. We get rid of the eigenvalue $1$, leaving ourselves with an update matrix $M_3$ whose dimension is equal to the dynamical dimension of the original stochastic matrix $M$. We show that in doing so, the linear dimension of the resulting sets in $\mathcal{T}_3$ is less than that of their counterparts in $\mathcal{T}_2$.

\begin{lemma}
\label{mctolds}
Let $M \in \mathbb{Q}^{(k+\ell)\times (k+\ell)}$ have nonzero eigenvalues $1, \lambda_1, \dots, \lambda_j$, such that the eigenvalue $1$ is simple, has multiplicity $\ell$, and is the only eigenvalue of modulus at least $1$. Let $\mu \in \mathbb{Q}^{k+\ell}$, and let $\mathcal{T} = \{T_1, \dots, T_h\}$ be a collection of semialgebraic sets. We can compute $n_0 \in \mathbb{N}$, $s \in \mathbb{Q}^{k+\ell}$, a rank-$k$ matrix $Q \in \mathbb{Q}^{(k+\ell) \times k}$, an invertible and non-degenerate matrix $A \in \mathbb{Q}^{k \times k}$ with eigenvalues $\lambda_1, \dots, \lambda_k$, a vector $v \in \mathbb{Q}^k$, and a collection of semialgebraic sets $\mathcal{T'} = \{T_1', \dots, T_h'\}$ such that:
\begin{enumerate}
\item For all $n$,
$
M^n \mu = s + QA^n v
$.
\item The eigenvalues of $A$ are $\lambda_1, \dots, \lambda_j$. Moreover, for any $i$, $y$ is a generalised order-$i$ eigenvector of $A$ corresponding to $\lambda \ne 0$ if and only if $Qy$ is a generalised order-$i$ eigenvector of $M$ corresponding to $\lambda$.
\item Let $n \ge n_0$. For every $i$, $M^n \mu \in T_i$ if and only if $A^n v \in T'_i$. If $T_i$ has linear and intrinsic dimensions $d_1, d_2$ respectively, then the corresponding dimensions of $T_i'$ are at most $d_1-1, d_2$.
\end{enumerate}
\end{lemma}
\begin{proof}
The proof of Lemma \ref{zeroelim} applies \emph{mutatis mutandis} to establish requirements (1) and (2). In this case, $P$ spans the eigenspace of the simple eigenvalue $1$, and we have $MP = P$ and $P'Q = O$. We choose $s, v$ from the unique representation of $\mu = s + Qv$, where $s$ is a linear combination of the columns of $P$.

Having established requirement (2), we deduce that the Euclidean norm $||A^n v||$ converges to $0$ exponentially quickly. We use this observation in computing $n_0, \mathcal{T}'$ as required. Let $U_1, \dots, U_h$ be the smallest subspaces containing $T_1, \dots, T_h$ respectively. We compute $\varepsilon \in \mathbb{Q}$ such that for all $i$, if $s \notin U_i$ then neither is any point $x$ in the Euclidean $\varepsilon$-neighbourhood of $s$. We compute $n_0$ such that for all $n \ge n_0$, $||A^n v|| = ||M^n \mu - s|| < \varepsilon$.

For all $i$, we define $T_i' = \{x \in \mathbb{R}^k : s + Qx \in T_i \land ||x|| < \varepsilon\}$. Observe that $T'_i$ is non-empty only if $s \in U_i$. It is also straightforward to adapt the reasoning from the proof of Lemma \ref{zeroelim}(2) to conclude that for $n \ge n_0$, $M^n \mu \in T_i$ if and only if $A^n v \in T_i'$ and that the intrinsic dimension of each $T'_i$ is at most that of the corresponding $T_i$. It remains to prove that if $T_i'$ is nonempty, then its linear dimension is less than that of $T_i$.

Consider $U_i$. Since $T'_i$ is non-empty, we know $s \in U_i$. By definition of a subspace, there exists a $(k+\ell - d)\times (k+\ell)$ matrix $H$ of rank $k+\ell -d$ such that $y \in U_i$ only if $Hy = \mathbf{0}_{k+\ell-d}$. Furthermore, since $s \in U_i$, we have that $Hs = \mathbf{0}_{k+\ell-d}$. We shall now identify a subspace that contains $T'_i$.

By definition, $x \in T'_i$ if and only if $s + Qx \in T_i$. Thus, $x \in T'_i$ only if $H \cdot Qx = \mathbf{0}_{k+\ell-d}$. Now, recall that the columns of $Q$ span $W$, the union of eigenspaces of the eigenvalues $\lambda_1, \dots, \lambda_j$, which is also the $k$-dimensional kernel of the rank-$\ell$ matrix $P'\in \mathbb{Q}^{\ell \times (k+\ell)}$. We know that $s$ is an eigenvector of $1$, and is hence not in $W$. Thus, we have at least one row $p'$ of $P'$ such that $p's \ne 0$. This row is thus necessarily linearly independent of the rows of $H$, for otherwise $Hs = \mathbf{0}_{k+\ell-d}$ would imply $p's = 0$. The fact that $P'Q = O$ by construction guarantees that $P' \cdot Qx = \mathbf{0}_\ell$.

The above argument allows us to append $P'$ to the rows of $H$, obtaining a matrix $G$ of rank at least $k+\ell -d+1$ by virtue of its rows $H, p'$ being linearly independent. We can assert a stronger claim: $x \in T'_i$ only if $G \cdot Qx = \mathbf{0}_{k+2\ell -d +1}$. By the rank-nullity theorem, it follows that $Qx$ must lie in a vector space of dimension $d' \le d-1$ that is contained in $W$, the span of the columns of $Q$.

Finally, we recall that $Q$ has rank $k$: thus the linear map $Q$ from $\mathbb{R}^k$ to $W$ is invertible via $Q'$. Thus, $Qx \in \text{span}(v_1, \dots, v_{d'}) \subseteq W$ (if and) only if \\ $x \in \text{span}(Q'v_1, \dots, Q'v_{d'})$.

This allows us to conclude that $x \in T$ only if $x$ is contained in a linear subspace of dimension at most $d-1$, and completes the proof of Lemma \ref{mctolds}.
\qed
\end{proof}

Finally, Lemma \ref{mctolds} completes the chain of reductions from $(\mu, M, \mathcal{T}, \mathcal{L})$ to $(\mu_2 = v, M_3 = A, \mathcal{T}_3, \mathcal{L}_3)$, where $A$ is invertible and has dimensions equal to the dynamical dimension of $M$. We have $\mathcal{L}_3 = \{\beta: \alpha(0)\cdots \alpha(n_0-1)\cdot \beta \in \mathcal{L}_2\}$, where $\alpha$ is the characteristic word of $M_2, \mu_1$ with respect to $\mathcal{T}_2$. We also have that the linear dimensions of sets in $\mathcal{T}_2$ are at less those of their correspondents in $\mathcal{T}$. This completes the proof of our decidability results.

\section{Proofs of the hardness results}
\label{hardness proofs}
In this section, we prove Theorems~\ref{ldstomc} and~\ref{skolemhard}, thus showing that our decidability results are tight, i.e.\ that further breakthroughs would entail the decidability of Skolem-$5$. Recall that a Markov chain is ergodic if its transition graph is irreducible (i.e.\ strongly connected) and aperiodic. Equivalently, for its stochastic matrix $M$ there exists an $n_0$ such that all entries of $M^{n_0}$ are strictly positive. In our hardness proofs, we shall construct $M$ with $n_0 = 1$.

\subsection{First hardness result}
We first show how to construct, given an arbitrary LDS $(v,A)$, an ergodic Markov chain $(\mu, M)$ that captures the dynamics of $(v,A)$.

\begin{corollary}[Of Proof of Lemma \ref{mctolds}]
\label{ldstomcbase}
Let $s \in \mathbb{Q}^{k+1}$ be a distribution
with strictly positive entries, $A \in \mathbb{Q}^{k \times k}$ and $v \in \mathbb{Q}^k$. We can compute an ergodic Markov chain $(\mu, M)$ and constants $\eta, \rho \in \mathbb{Q}$ such that $M^n \mu = s + \eta \rho^n \begin{bmatrix} I \\ -\mathbf{1}_{k}^\top\end{bmatrix}A^n v$ for all $n$.
\end{corollary}
\begin{proof}
Choose $Q$ from Lemma \ref{mctolds} to be $\begin{bmatrix} I \\ -\mathbf{1}_{k}^\top\end{bmatrix}$. Observe that $s$ is linearly independent of the columns of $Q$: indeed, we have $\mathbf{1}_{k+1}^\top Q = \mathbf{0}_{k+1}^\top$, but $\mathbf{1}_{k+1}^\top s = 1$. Thus the matrix $R = \begin{bmatrix}s & Q\end{bmatrix}$ is invertible. Moreover, we can easily check that its inverse ${R}^{-1}$ is of the form $\begin{bmatrix}\mathbf{1}_{k+1}^\top \\ Q'\end{bmatrix}$.

Now we choose $\eta$ such that the magnitude of the largest entry of $\eta Qv$ is less than the smallest entry of $s$. We take $\mu = s + \eta Qv$, and observe that all entries of $\mu$ are positive, and moreover that $\mathbf{1}_{k+1}^\top \mu = 1$, making $\mu$ a distribution. We choose $\rho$ to ensure that the magnitude of the largest entry of $\rho QAQ'$ is smaller than the smallest entry of $s$.

Now, we take
\begin{align*}
M &= RDR^{-1} =
\begin{bmatrix}
s & Q
\end{bmatrix}
\begin{bmatrix}
1 & \mathbf{0}_k^\top \\
\mathbf{0}_k & \rho A
\end{bmatrix}
\begin{bmatrix}
\mathbf{1}_{k+1}^\top \\
Q'
\end{bmatrix}
=
\begin{bmatrix}
s & Q
\end{bmatrix}
\begin{bmatrix}
\mathbf{1}_{k+1}^\top \\
\rho AQ'
\end{bmatrix}
= s\cdot \mathbf{1}_{k+1}^\top + \rho QAQ'.
\end{align*}
The choice of $\rho$ ensures that each entry of $M$ is positive, and it is moreover easy to check that $\mathbf{1}_{k+1}^\top M = \mathbf{1}_{k+1}^\top$. This makes $M$ a stochastic matrix.
Since all entries of $M$ are positive, $(\mu, M)$ is ergodic.
It is then straightforward to verify that $M^n \mu = s + \eta\rho^n QA^n v$.
\qed
\end{proof}

We mention that Cor.~\ref{ldstomcbase} generalises the result of \cite{vahanwalaergodic}.
We now prove Theorem \ref{ldstomc}. Given an instance $(v, A, \mathcal{T}, \mathcal{L})$ of the model-checking problem for linear dynamical systems with $\mathcal{T} = \{T_1, \dots, T_h\}$ being a collection of homogeneous sets, we apply Corollary \ref{ldstomcbase} to obtain $\mu, M$, and consider $Q'$ from its proof. We note that by homogeneity, for each $i$, $A^n v \in T_i$ if and only if $\eta \rho^n A^n v \in T_i$. We construct $\mathcal{T}' = \{T_1', \dots, T_h'\}$ as $T_i' = \{y: \mathbf{1}_{k+1}^\top (y-s) = 0 \land Q'(y-s) \in T_i\}$. The first conjunct ensures that $y-s$ is in the space spanned by the columns of $Q$. The set $T_i'$ is $s$-homogeneous by construction, and it is straightforward to check that $M^n \mu \in T_i'$ if and only if $\eta \rho^n A^n v \in T_i$ if and only if $A^n v \in T_i$. The equivalent instance $(\mu, M, \mathcal{T}', \mathcal{L}')$ is obtained by taking $\mathcal{L}'$ to be the $\omega$-regular language obtained from $\mathcal{L}$ by simply renaming letters in $2^{\mathcal{T}}$ to their counterparts in $2^{\mathcal{T}'}$.

\subsection{Second hardness result}
In this section, we take a hard instance $(A, v, h)$ of Skolem-$5$ and construct an equivalent instance $(\mu, M, T)$ of the reachability problem where $(\mu, M)$ is an ergodic Markov chain and $T$ is a semialgebraic target with intrinsic dimension at most~$2$.
We will need the following.

\begin{corollary}[Of Proof of Lemma \ref{mctolds}]
	\label{cor::making-ergodic-mc-from-lds}
	Let $s \in \mathbb{Q}^{k+1}$ be a distribution with strictly positive entries, $B \in \mathbb{Q}^{k \times k}$, and $v \in \mathbb{Q}^k$.
	Write $Q = \begin{bmatrix} I \\ -\mathbf{1}_{k}^\top\end{bmatrix}$, and suppose that every entry of $Q^{-1}BQ$ is less than every entry of $s$ in magnitude.
	Then we can compute an ergodic Markov chain $(\mu, M)$ and $\eta \in \rat_{>0}$ such that $M^n\mu = s + \eta QB^nv$ for all $n$.
\end{corollary}
\begin{proof}
	Same as the proof of Cor.~\ref{ldstomcbase}, with the difference that $\rho A$ is replaced with~$B$, whose entries are already assumed to be sufficiently small.
	\qed
\end{proof}

Let $(A, v, h)$ be a hard instance of Skolem-$5$. We have to decide whether there exists $n$ such that $u_n = 0$, where $u_n = h^\top A^n v$ is a rational LRS.
As discussed in~\cite[Sec.~2.3]{karimovthesis}, we can assume the eigenvalues of $A$ are of the form $\lambda, \overline{\lambda}, \gamma, \overline{\gamma}, \rho$ and satisfy $|\lambda| = |\gamma| > |\rho| > 0$.
By considering the sequences $\seq{h^\top \bigl(A^2\bigr)^nv}$ and $\seq{h^\top \bigl(A^2\bigr)^n(Av)}$ separately if necessary, we can further assume $\rho > 0$.
Since the eigenvalues are the roots of characteristic polynomial of $A$, we have that $|\lambda|^4 \rho = \lambda \overline{\lambda} \gamma \overline{\gamma} \rho$ is rational.
Furthermore, $\rho$ must also be rational.
To see this, let $\sigma$ be a Galois automorphism of the splitting field $\mathbb{Q}(\lambda, \bar\lambda, \gamma, \bar\gamma, \rho)$ of the characteristic polynomial of $A$.
We have that $\sigma$ permutes $\lambda, \bar\lambda, \gamma, \bar\gamma, \rho$.
Moreover, $\lambda\bar\lambda = \gamma\bar\gamma$, and hence $\sigma(\lambda)\sigma(\bar\lambda) = \sigma(\gamma)\sigma(\bar\gamma)$. If $\sigma$ were to permute any of the other eigenvalues to $\rho$, this equality would be violated. Thus, $\rho$ is fixed by every automorphism and hence is rational.
We conclude that $|\lambda|^4$ must also be rational.

Let $s$ be a distribution with strictly positive entries, and
\[
\kappa = \frac{1}{|\rho|} \cdot \frac{|\rho|^{4c}}{|\lambda|^{4c}} \in \rat
\] for $c$ sufficiently large so the condition of Cor.~\ref{cor::making-ergodic-mc-from-lds} is met with $B \coloneqq \kappa A$.
Then the eigenvalues $\kappa \rho, \kappa \lambda$ of $B$ satisfy $(\kappa |\rho|)^{4c-1} = (\kappa |\lambda|)^{4c}$.

We will next construct an instance $(v, B, \widetilde{T})$ of the model-checking problem such that $\widetilde{T}$ is a semialgebraic set of intrinsic dimension at most $2$ and
\[
h^\top A^nv = 0 \Leftrightarrow B^n v \in \widetilde{T}
\]
for all $n$.
Write $r = |\lambda|$, $H = \{x \st h^\top x = 0\}$, and $B = PJP^{-1}$, where $J$ is in real Jordan form.
Without loss of generality, we can assume
\[
J = \operatorname{diag}(\kappa r  \Lambda, \kappa r  \Gamma, \kappa\rho)
\] where $\Lambda, \Gamma$ are $2 \times 2$ rotation matrices.\footnote{This is easily seen: let $v_\lambda, \bar v_\lambda, v_\gamma, \bar v_\gamma, v_\rho$ be eigenvectors of $\kappa A$, and take the columns of $P$ to be $(v_\lambda + \bar v_\lambda), i(v_\lambda - \bar v_\lambda), (v_\gamma + \bar v_\gamma), i(v_\gamma - \bar v_\gamma), v_\rho$.}
Write $Pv$ as $(v_1,\ldots,v_5)$, $d_1 = \sqrt{v_1^2+v_2^2}$, and $d_2 = \sqrt{v_3^2+v_4^2}$.
We have
\begin{align*}
	\Vert (\kappa r \Lambda)^n (v_1,v_2) \Vert &= d_1 \cdot (\kappa  r)^{n} = d_1 \cdot ((\kappa  r)^{n/(4c-1)})^{4c-1}\\
	\Vert (\kappa r \Gamma)^n (v_3,v_4) \Vert &= d_2 \cdot (\kappa  r)^{n} = d_2  \cdot((\kappa  r)^{n/(4c-1)})^{4c-1}\\
	v_5 \cdot (\kappa \rho)^n  &= (\kappa  r)^{4cn/(4c-1)}v_5 = v_5((\kappa  r)^{n/(4c-1)})^{4c}
\end{align*}
It follows that $\seq{J^nPv}$ is contained in the set $S$ of all $(x_1,\ldots,x_5) \in \rel^d$ satisfying the following equations.
\begin{align*}
	(x_1^2+x_2^2)d_2^2 &= (x_3^2+x_4^2)d_1^2\\
	((x_1^2+x_2^2)v_5^2)^{4c} &= x_5^{4c-1} \cdot v_5 d_1^{8c}.
\end{align*}
On the other hand, let $x(t) = \bigl(\frac{1-t^2}{1+t^2},\frac{2t}{1+t^2}\bigr)$, and recall that the unit circle can be parametrised as $\{(-1,0)\} \cup \{x(t)\st t\in \rel\}$.
We have that $S = S_1 \cup \cdots \cup S_4$, where
\begin{align*}
	S_1 &= \{(d_1 s^{4c-1} \cdot x(t_1), \, d_2  s^{4c-1} \cdot x(t_2),\, v_5 s^{4c}) \st s \ge 0, t_1,t_2 \in \rel \}\\
	S_2 &= \{( d_1 s^{4c-1} \cdot (-1, 0), \,d_2  s^{4c-1} \cdot x(t_2),\, v_5 s^{4c}) \st s \ge 0, t_2 \in \rel \}\\
	S_3 &= \{(d_1 s^{4c-1} \cdot x(t_1), \,d_2  s^{4c-1} \cdot (-1, 0),\, v_5 s^{4c}) \st s \ge 0, t_1 \in \rel \}\\
	S_4 &= \{( d_1 s^{4c-1} \cdot (-1, 0), \,d_2  s^{4c-1} \cdot (-1, 0),\, v_5 s^{4c}) \st s \ge 0\}.
\end{align*}
Each of $S_1,\ldots,S_4$ is parametrised by a Zariski-continuous function.
The domains $\rel_{>0} \times \rel^2, \rel_{>0} \times \rel, \rel_{>0}$ are all irreducible, and have dimensions $3,2,2,1$, respectively.
Hence each $S_i$ is irreducible, and $\dim(S_1) \le 3$, $\dim(S_2),\dim(S_3) \le 2$, and $\dim (S_4) \le 1$.
It can also be shown that $P\cdot H$, which is also irreducible, is not contained in any $S_i$ and vice versa.
Hence $\dim(P\cdot H \cap S_i) \le 2$ and therefore $\dim(P\cdot H \cap S) \le 2$.
We can therefore define $\widetilde{T} = H \cap P^{-1} S$.

Applying Cor.~\ref{cor::making-ergodic-mc-from-lds}, we compute an ergodic Markov chain $(\mu, M)$ and $\eta \in \rat_{>0}$ such that $M^n \mu = s+ \eta Q B^n v$.
Define $T = s+\eta Q \widetilde{T}$.
Then $\dim(T) \le 2$, and
\[
M^n\mu \in T \Leftrightarrow B^nv \in \widetilde{T} \Leftrightarrow A^nv \in H.
\]
This completes the proof of Theorem~\ref{skolemhard}.

\section{Discussion}
\label{conclusion}
We conclude by offering perspective on our decidability results and techniques. We observe that our reduction from Markov chains to ordinary linear dynamical systems becomes more powerful as the difference between the order (number of states) and the dynamical dimension of the Markov chain increases. Recall that this difference is attributed to the multiplicity of the eigenvalues of the underlying stochastic transition matrix $M$ that have modulus $0$ or $1$.

Having the eigenvalue $0$ with multiplicity $\ell$ has a ``dependency'' effect: all distributions $M^\ell \mu, M^{\ell+1}\mu, \dots$ will satisfy $\ell$ homogeneous linear constraints of the form $h_1 x_1 + \dots + h_k x_k = 0$, where $x_i$ is the probability of being in state $i$.

Eigenvalues of modulus $1$ other than $1$ itself are necessarily roots of unity, and hence describe the ``periodic'' behaviour of the system. They occur if every cycle in a bottom strongly connected component $G$ of the graph has length divisible by some $c > 1$: $G$ can then be partitioned into $G_0, \dots, G_{c-1}$ such that for each $i$, all transitions starting in $G_i$ end in $G_{i+1}$. Thus, the probability mass that lands in $G$ is also cycled between these partitions. By taking $M^c$ to be the transition matrix, we turn each partition into a separate bottom strongly connected component.

The eigenvalue $1$ accounts for the ``stationary'' behaviour of the system, and occurs as many times as there are bottom strongly connected components in the graph of the transition system. This is intuitively because every bottom strongly connected component $G$ corresponds to a unique stationary distribution $\mu_G$ which assigns nonzero probability to the states in $G$ and zero probability to all other states. A given initial distribution determines how the probability mass will eventually be distributed between the bottom strongly connected components. Furthermore, if these are aperiodic, then within each component $G$, the mass that lands in $G$ will inevitably distribute itself in proportion to $\mu_G$

Our reduction extracts the dynamics of the stochastic system by making dependencies implicit, partitioning periodicity into phases, and subtracting the stationary behaviour. The sequence of vectors in the resulting linear dynamical system converges to the origin exponentially quickly. In practice:
\begin{enumerate}
\item Markov chains might have strong dependencies, be highly periodic, and have several bottom strongly connected components. This could make the dynamical dimension for the model-checking problem small.
\item The sets in $\mathcal{T}$ might not have the limiting distributions on their boundaries, making the resulting sets in $\mathcal{T}'$ empty, and the resulting model-checking problem trivially decidable.
\end{enumerate}

The optimistic practitioner would consider common systems and believe that model checking Markov chains as distribution transformers ought to be easier than arbitrary linear dynamical systems; while a skeptical theoretician would recall ergodic Markov chains and expect that the problems are essentially equivalent. By efficiently distilling the dynamics of Markov chains from their dependencies, periodic nature, and limiting behaviour and quantifying their respective contributions, we formally reconcile these opposing intuitions.

\bibliographystyle{plain}
\bibliography{main}
\end{document}

%% file: macros.tex
\newcommand{\rel}{\mathbb{R}}
\newcommand{\rat}{\mathbb{Q}}

\newcommand{\seq}[1]{( #1 )_{n=0}^{\infty}}
\newcommand{\st}{\colon}